\documentclass[prl,twocolumn,showpacs,superscriptaddress,amsmath,amssymb]{revtex4}
\usepackage[utf8]{inputenc}
\usepackage{amsthm}
\usepackage[author={Ariel}]{pdfcomment}

\usepackage[ruled]{algorithm}
\usepackage{algpseudocode}
\usepackage{graphicx,xspace}
\usepackage[export]{adjustbox}
\usepackage[draft]{fixme}
\algnotext{EndFor}
\algnotext{EndIf}

\usepackage{enumitem}
\usepackage[final]{changes}
\newcommand{\subscript}[2]{$#1 _ #2$}

\newcommand{\NN}{\mathbb{N}}

\newtheorem{lemma}{Lemma}
\newtheorem{definition}{Definition}

\newtheorem{fact}{Fact}
\newtheorem{proof-sketch}{Proof Sketch}

\usepackage[textwidth=0.6in,textsize=scriptsize]{todonotes}

\begin{document}
\title{Non-signaling deterministic models for non-local correlations have to be uncomputable}

\author{Ariel Bendersky}
\affiliation{Departamento de Computaci\'on, FCEN, Universidad de Buenos Aires, Buenos Aires, Argentina}
\affiliation{CONICET, Argentina}
\author{Gabriel Senno}
\affiliation{Departamento de Computaci\'on, FCEN, Universidad de Buenos Aires, Buenos Aires, Argentina}
\affiliation{CONICET, Argentina}
\author{Gonzalo de la Torre}
\affiliation{ICFO-Institut de Ciencies Fotoniques, The Barcelona
Institute of Science and Technology, 08860 Castelldefels,
Barcelona, Spain}
\author{Santiago Figueira}
\affiliation{Departamento de Computaci\'on, FCEN, Universidad de Buenos Aires, Buenos Aires, Argentina}
\affiliation{CONICET, Argentina}
\author{Antonio Acin}
\affiliation{ICFO-Institut de Ciencies Fotoniques, The Barcelona
Institute of Science and Technology, 08860 Castelldefels,
Barcelona, Spain}
\affiliation{ICREA, Pg. Lluis Companys 23, 08010 Barcelona, Spain}

\begin{abstract}
 Quantum mechanics postulates random outcomes. However, a model making the same output predictions but in a deterministic manner would be, in principle, experimentally indistinguishable from quantum theory. In this work we consider such models in the context of non-locality on a device independent scenario. That is, we study pairs of non-local boxes that produce their outputs deterministically. It is known that, for these boxes to be non-local, at least one of the boxes' output has to depend on the other party's input via some kind of hidden signaling. We prove that, if the deterministic mechanism is also algorithmic, there is a protocol which, with the sole knowledge of any upper bound on the time complexity of such algorithm, extracts that hidden signaling and uses it for the communication of information.
\end{abstract}


\maketitle


Bell nonlocality~\cite{bell1964einstein} makes us choose between determinism and the non-signaling principle~\cite{valentini}. That is, if one wants to account in a deterministic manner for the non-local correlations that quantum mechanics predicts and which we are now almost certain~\cite{hensen2015loophole,giustina2015significant,PhysRevLett.115.250402} that Nature exhibits, one must allow for the existence of some kind of signaling mechanism that links distant measurement choices and outcomes. But, since quantum correlations are non-signaling, such signaling mechanism must be restricted to the so-called hidden variables, and not reach the phenomenological level.
Known examples of deterministic non-local theories violating the non-signaling principle (also referred to as \textit{parameter independence}~\cite{shimony1986events}) at the hidden-variable level are: the hidden variable model with communication of Toner and Bacon~\cite{toner2003communication} and, more prominently, Bohmian mechanics~\cite{bohm1952suggested}. For those models that use classical communication to mimic non-locality, one can in fact study the amount of communication needed (see, for example,~\cite{regev2009simulating,shi2008tensor,degorre2011communication}).

A reasonable feature that one would expect of any physical model is that it is computable~\cite{feynman1982simulating}. This means that, in principle, one should be able to write a computer program that given a description of an experiment (that is, the measurement choices and the state of the system) outputs the model's outcomes predictions (these being probabilities in the case of quantum mechanics).

Our main result is that, on the contrary, deterministic models of non-local correlations need to be uncomputable if we want to prevent those correlations from being signaling. In other words, we show that if the deterministic model is computable, the hidden signaling mechanism used to exhibit non-locality can be extracted at the observation level and used for the communication of information. More specifically, we give a protocol to perform one-way communication between two observers holding computable non-local boxes. 

There are a few previous results on this direction. First, this result has a flavour similar to \cite{yurtsever2000quantum}. However, we obtain our result in a device-independent scenario, that is, without assuming quantum mechanics, and provide an explicit communication protocol. Second, in \cite{PhysRevA.92.052102, baumeler2016causality} it is shown that some non-local boxes fed with algorithmically random strings can't produce computable outputs. We show that any set of non-local boxes that work by using hidden communication and an algorithm to define their outputs, can be used to signal at the phenomenological level.

This paper is organized as follows: first we introduce the scenario that we are considering. Then we briefly review the tools from computer science that we need to resort to in order to prove our main result. Finally, we present and prove our results.

\begin{figure}[ht]
\adjincludegraphics[width=6cm,trim={0 0 0 {.05\width}},clip]{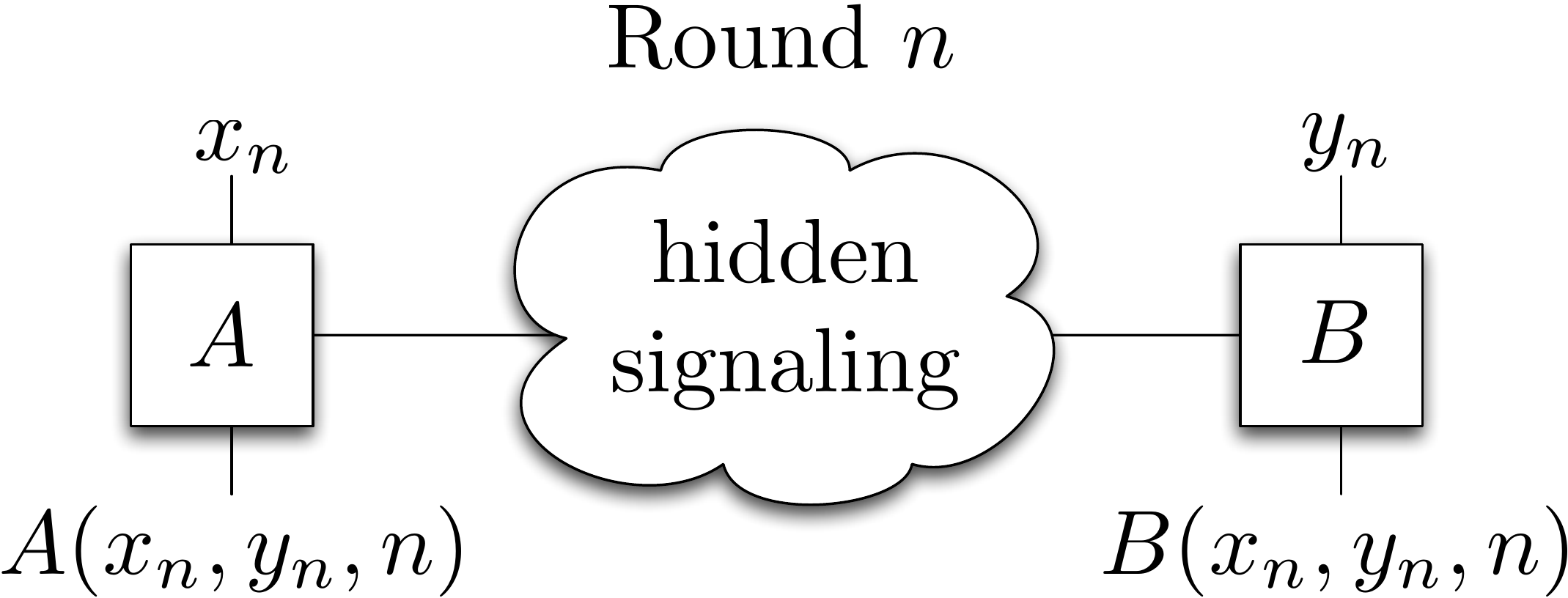}
\caption{Schematic representation of the scenario considered. Two distant observers, Alice and Bob, run a Bell test by implementing measurements on two systems. The observed correlations are described by a hidden-signaling mechanism plus computable functions determining the outputs given the inputs at each round $n$.}
\end{figure}

\textit{The scenario. -- }We consider a standard Bell scenario. For the sake of simplicity, we present our results for the simplest Clauser-Horne-Shimony-Holt (CHSH) Bell test~\cite{chsh} where we have two parties, Alice and Bob, each one with a box that has a binary input and a binary output, but the extension to other scenarios is straightforward.

\begin{definition}\label{def:non-local-pair} A pair of boxes $A,B$ with inputs $(x_i)_{i\in\NN}$ and $(y_i)_{i\in\NN}$ that are independent tosses of a fair coin, and outputs $(a_i)_{i\in\NN}$ and $(b_i)_{i\in\NN}$ is \textit{non-local} iff 
$$
p(a,b|x,y)=\lim_{n\to\infty} \frac{4\cdot \#\{i< n \mid(x_i,y_i,a_i,b_i)=(x,y,a,b)\}}{n}
$$
violates a Bell inequality with probability 1.
\end{definition}


Our goal is to study deterministic and computable models that reproduce non-local correlations. Thus, the boxes under consideration are computable, meaning that there is a computable function $A(x,y,n)$ which gives, for each round $n$, the output of Alice's box when her input is $x$ and Bob's input is $y$, and a similar function $B(x,y,n)$ for Bob's box. Definition \ref{def:non-local-pair}, however, is general enough to cover the usual non-deterministic scenario as well. 

As we said in the introduction, because we are looking at deterministic boxes generating non-local correlations, their outputs have to depend on each other's input. Since the boxes are computable, this is the only information they need to share, as any other necessary data can be computed from the inputs. It is important to note that, although it seems that our toy model is signaling, and therefore it wouldn't come as a surprise that Alice can signal to Bob, this is not the case. The model uses signaling for its internal workings but does not necessarily allow Alice and Bob to send information to each other. For instance, if one does not impose the computable condition to $A$ and $B$, one can easily simulate quantum mechanics in a way that is completely equivalent and indistinguishable from standard quantum theory when having access only to the boxes inputs and outputs and local randomness. The possibility of having non-locality at the observational level depends also on the computable condition.

It is easy to see that, if the dependence between distant inputs and outputs happens in only finitely many rounds, the boxes are essentially local. So, we have that:

\begin{lemma}\label{lemma1} If $A,B$ 
is a non-local pair, 
then
$\exists x\in\{0,1\}. A(x,0,n)\neq A(x,1,n)$ or $\exists y\in\{0,1\}. B(0,y,n)\neq B(1,y,n)$ for infinitely many $n$s.
\end{lemma}

In the following, $A$ and $B$ will form a non-local pair and, without loss of generality, we make the next assumption:
\begin{enumerate}[label=(*)]
\item\label{item:assumption} At least Bob's function $B$ depends on Alice's input for infinitely many $n$s.
\end{enumerate}

Lemma \ref{lemma1} tells us that, for infinitely many $n$s, the value of $x$ can be determined
from the output of $B$ with the suitable choice of $y$.en el azul de la última 
Therefore, if Alice and Bob knew how to compute $B$, they could trivially signal. The situation we want to study is when $A$ and $B$ are unknown. What we show next is that, with the assumption that $A$ and $B$ are computable functions, one can actually devise a protocol to transmit one-way information from Alice to Bob with the sole knowledge of some upper bound on their time computational complexity. Before showing the protocol, we need to introduce some concepts from computability theory.


Our protocol is based on two main concepts, namely learnability schemes of computable functions and a notion of randomness against adversaries with bounded computational power. We explain these two concepts in what follows.

\textit{Learnability of computable functions. --} The main ingredient in the protocol we are about to describe is that of \textit{learning} computable functions from a finite number of samples~\cite{zeugmann2008learning}. That is, we will provide Bob with a Turing machine $L_t$ such that, on input (some coding of) $(x_1,y_1,B(x_1,y_1,1))\dots(x_n,y_n,B(x_n,y_n,n))$, for a sufficiently large $n$, it outputs (the index of) some Turing machine that computes $B$. From that $n$ onwards, Bob will have a way to guess Alice's input (see Lemma \ref{lemma1}). We call $L_t$ a \textit{learner} for the class of functions computable in time $O(t)$. It is important to notice that such $n$ after which the function is learned is not computable (i.e., it is impossible to know when the function has already been learned).

The behaviour of $L_t$ consist of enumerating all machines that run in time $O(t)$ until finding the first machine which reproduces the input-output behaviour of the target function seen so far. Once such a machine is found, its index is returned. Since, by assumption, the target function is computed in time $O(t)$, after finitely many mistakes, the learner will output the index of one such machine. See Fig. \ref{figLearn} for schematic description. This algorithm is a special case of a more general technique called \textit{learning by enumeration}~\cite{gold1967language}. As its name suggests, this technique works for every class of machines which can be computably enumerated, the class of machines which run in $O(t)$ being a particular example.

\begin{figure}[ht]
\includegraphics[scale=.25]{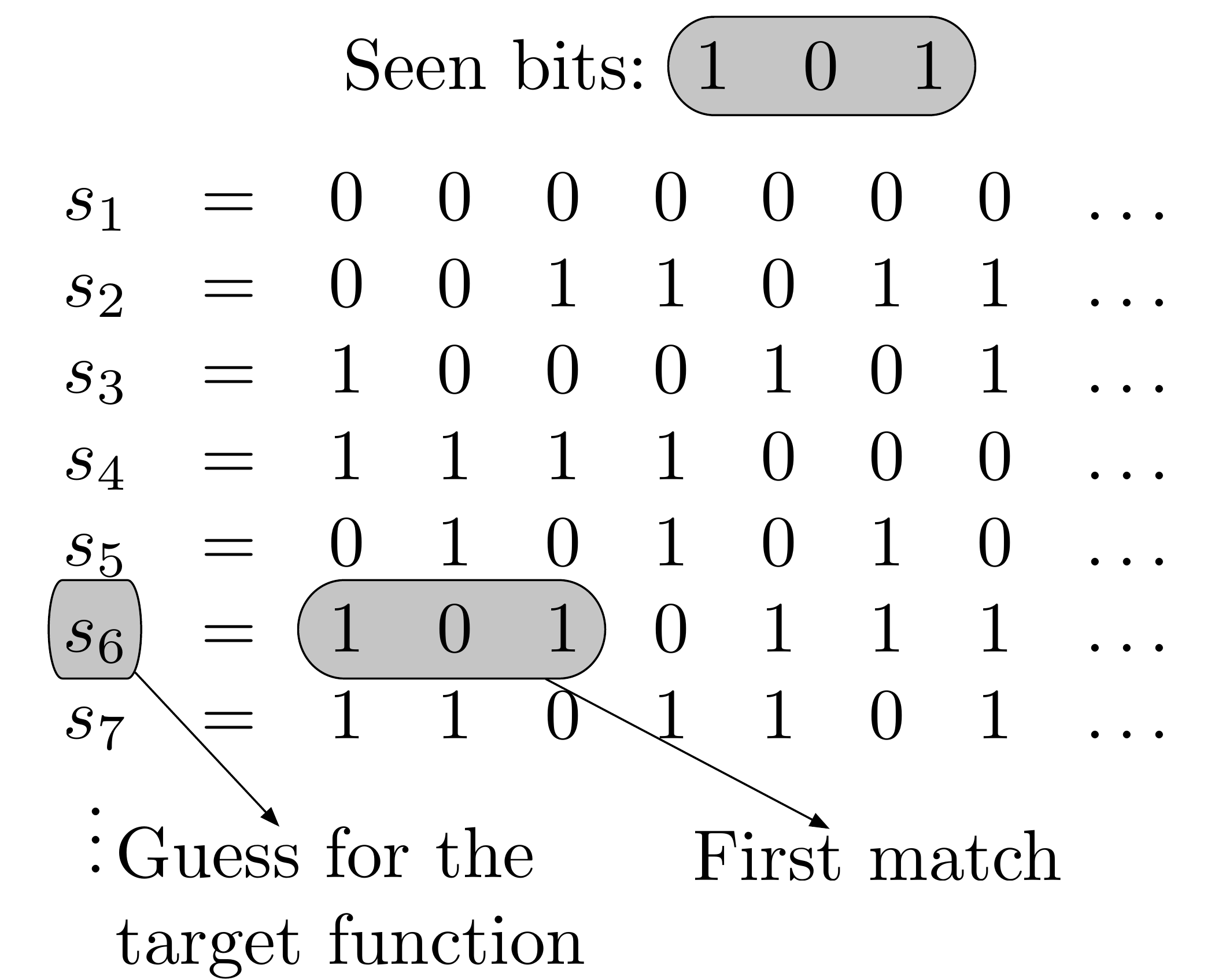}
\caption{Suppose $\{s_i\}_{i\in\NN}$ is a (computable) enumeration of Turing machines which run in time $O(t)$. The $i$-th row represents the sequence $s_i(0),s_i(1),s_i(2),\dots$.
After seeing $f(0)=1$, $f(1)=0$ and $f(2)=1$, the guess for the target function is defined as (the index of) the first machine whose outputs match those values (in the example, the guess is $s_6$). If the guess is correct, it is possible to predict the whole tail of the sequence in advance. \label{figLearn}}
\end{figure}

\textit{$t$-randomness. --} It is a basic result from computability theory that computers, or any device equivalent to a Turing machine, cannot generate random sequences. However, one can consider a notion of {\em t-randomness}~\cite{downey2010algorithmic,nies2009computability}, in which the degree of randomness is defined with respect to an adversary whose computing time is bounded by some computable function $t$. Intuitively, the notion of a sequence random with respect to a time bound $t$ will be related to the impossibility of the adversary to \textit{predict} its symbols using a machine running in time $O(t)$. More precisely, let $\Sigma=\{b_1,\dots,b_{k}\}$ be a finite alphabet and $\Sigma^\omega$ the set of infinite sequences of symbols in $\Sigma$. A sequence $X\in\Sigma^\omega$ is $t$-random if and only if there is no betting strategy, computable in deterministic time $O(t)$, to win unbounded money betting on its symbols, one by one, using the information of the symbols already seen. 

The betting game is as follows: the gambler starts with some positive initial capital $M(\emptyset)>0$. After seeing the first $n$ symbols of $X$, which we note $X\upharpoonright n$, he has some capital $M(X\upharpoonright n)$, and bets some fraction $d_i$ of it on the next symbol being $b_i$ (for $i=1\dots k$). After that, the next symbol, namely $X(n)$, is revealed, and his new capital becomes
$$
M(X\upharpoonright n+1)=M(X\upharpoonright n)\left(1+k\cdot d_{X(n)} -\sum_{i}d_i\right),
$$
that is, he looses the money for the bets to symbols different from $X(n)$ and wins $k$ times the amount bet on $X(n)$. The gambler wins the game if his capital grows unboundedly as $n$ goes to infinity.

The following two facts about $t$-random sequences will be useful for what comes next:

\begin{fact}\label{fact-random} Let $X\in\Sigma^\omega$, let $\Gamma$ be a non-trivial subset of $\Sigma$, and let $g:\NN\to\{0,1\}$ be computable in time $O(t)$ with $t=\Omega(n^2)$ such that:
 \begin{itemize}
 \item for almost all $n$, if $g(n)=1$ then $X(n)\not\in\Gamma$, and
 \item for infinitely many $n$, $g(n)=1$.
 \end{itemize}
 Then $X$ is not $t$-random.
\end{fact}

\begin{proof}[Proof sketch.]

The idea is that, if for infinitely many $n$ you know that the next symbol will be different from the ones in $\Gamma$, waiting for those positions to bet to the other symbols is a winning strategy. More formally, let $m$ be such that for all $n>m$ if $g(n)=1$ then $X(n)\not\in\Gamma$.  It is easy to see that betting evenly on all symbols when $g(n)=0$ or $n\leq m$, and, for all $n>m$, betting $M(X\upharpoonright n)/\#(\Sigma \setminus \Gamma)$ on each symbol not in $\Gamma$ when $g(n)=1$ makes our capital grow unboundedly. This strategy is computable in deterministic time $O(t)$, because $g$ is. Then, $X$ is not $t$-random.
\end{proof}

\begin{fact}[See e.g.~\cite{figueira2015feasible}]
Given a program for the time function $t$, one can compute a $t$-random sequence in deterministic time $O(t(n)\cdot log(t(n))\cdot n^3)$.
\end{fact}

\textit{The signaling protocol. --} We are now in position to present our main result: the construction of a protocol that would allow two parties sharing non-local correlations to signal if the mechanism reproducing these correlations was computable and time bounded by a function $t$ known to Alice and Bob. 

The key idea of the protocol will be for Bob to perform a learnability scheme on the outputs of his box so as to, after finitely many rounds, be able to guess the future outputs and use them to tell Alice's input (see Lemma \ref{lemma1}). There are two issues that we will need to deal with in this approach:

\begin{itemize}
\item in order for Bob to learn a program to compute the function $B$, he needs to know Alice's inputs $x$, but the whole idea of this protocol was that these were conveyed from her to him trough the interaction with the boxes.
\item Bob will not be able to tell when he has effectively learned $B$.
\end{itemize}

The trick that will allow us to cope with these two issues is, to our knowledge, a new connection between $t$-randomness and learnability theory which consist on randomnly alternating between two kinds of rounds: learning rounds and signaling rounds. The former are rounds in which both parties run the Bell test using predetermined inputs known to both Alice and Bob so that Bob can learn the deterministic function $B$ determining his observed outputs. In principle, running a Bell test with predetermined inputs is not possible, but the idea is to use a predetermined sequence that is $t$-random, that is effectively random for a process computable in time O(t), such as the functions $A$ and $B$ determining the outputs in each device. The signaling rounds are used by Alice to signal the message to Bob, assuming $B$ is already known. In what follows we provide a more detailed description of the protocol and a proof of its soundness. 


First of all, Alice and Bob choose a computable function $t$ and assume $B$ is computable in deterministic time $O(t)$ (the protocol will fail if this assumption is false). As we said before, Bob will be using a learner $L_t$ for the class of functions computable in deterministic time $O(t)$. On the learning rounds, Alice and Bob will input their boxes with a prearranged input pair and Bob will use the output of his box to, through $L_t$, update his guess for a program that computes $B$. On the signaling rounds, Alice will input her message and Bob, acting according to his current  guess for $B$, will choose, whenever possible, the input $y$ that allows him to tell Alice's input $x$.

The protocol has thus three parameters: a computable time function $t$, a sequence
$$
S\in \{(0,0),(0,1),(1,0),(1,1),1,\dots,m\}^\omega
$$
(which is the one shared by Alice and Bob to perform the switching between the two kinds of rounds), plus a number $m$ which represents the size of the message that Alice wants to send to Bob.

All in all, here is the signaling protocol $\mathcal{P}(t,S,m)$:

On each round $n$:
\begin{enumerate}
	\item \textbf{Learning round: }if $S(n)=(x,y)$, Alice inputs $x$ and Bob inputs $y$. Furthermore, Bob sets his current guess $\widetilde B$ of a Turing machine that computes $B$ to $L_t((x_{i_1},y_{i_1},B(x_{i_1},y_{i_1},i_1))\dots(x,y,B(x,y,n)))$, with $i_k$ being the past learning rounds.
	\item \textbf{Signaling round}: if $S(n)=i\in\{1,\dots,m\}$, Alice inputs the $i$th bit of her message and Bob uses his current guess $\widetilde B$ of a program that computes $B$ to see if there is a $y$ such that $\widetilde B(0,y,n)\neq \widetilde B(1,y,n)$. If there is such $y$, he inputs it. If not, he inputs $0$.
\end{enumerate}

For this protocol to be sound, it suffices that the following properties hold:

\begin{enumerate}[label=(\subscript{P}{\arabic*})]
\item\label{item:prop1} There exists a number of round $n$ such that for all $m\geq n$, and $x,y\in\{0,1\}$, we have $\widetilde B(x,y,m)=B(x,y,m)$, i.e.\ the learning process  converges to $B$.
\item\label{item:prop2} For the $i$-th bit of Alice's message and for infinitely many $n$, $S(n)=i\in\NN$ and $\exists y\in\{0,1\}. B(0,y,n)\neq B(1,y,n)$, i.e.\ the signaling process works for infinitely many rounds.
\end{enumerate}

It is clear that these two properties give us signaling correlations because by \ref{item:prop1} after finitely many rounds the program that Bob uses in the signaling rounds correctly computes $B$ and, by \ref{item:prop2}, the number of rounds in which he will be able to use such program to tell every bit of Alice's message is infinite.

At this point, we can further clarify the assumption on the computational complexity of $B$. Our protocol is based on the existence of a learner for the class of computable functions to which we assume $B$ belongs. One would like to use a learner for the class of \textit{all} computable functions but it is a fundamental result in computability theory~\cite{zeugmann2008learning} that this class is not learnable. Of course we could have restricted the class of functions in some other way. For instance, instead of having a bound in the computational time, one could have bound the computational space.

Now, whether \ref{item:prop1} and \ref{item:prop2} hold or not will depend on the choice of shared switching sequence $S$. For example, if the $S(n)$ are independent and uniformly distributed random variables, it is easy to see that \ref{item:prop1} and \ref{item:prop2} hold with probability $1$. But this would make the argument too weak, as it would mean that Alice and Bob have access to a non-computable (random) sequence to test models of nature that are assumed to use only computable functions. On the other hand, it is not hard to see that if $S$ is chosen such that, for example, it indicates learning in the odd rounds and signaling in the even, the learning could converge to a program that coincides with $B$ in almost all odd positions but, for the even positions, it outputs, say, the negation of $B$ (this program, of course, also runs in time $O(t)$). One can then expect that some notion of computable randomness is needed for the protocol to work. The question is: can we find a computable sequence $S$ that does the job?

In general, no easily predictable sequence $S$ is suitable. The idea will be to define $S$ computationally hard enough to predict (this will be related to the complexity that the protocol assumes on~$B$). The existence of such $S$ will come from the theory of \textit{computable randomness}. Lemmas \ref{lemma-learning} and \ref{lemma-signaling} below say that, when 
$$
S\in \{(0,0),(0,1),(1,0),(1,1),1,\dots,m\}^\omega
$$
is $t$-random, the protocol $\mathcal{P}(t,S,m)$ is sound.

\begin{lemma}\label{lemma-learning}
If $S$ is $t$-random then $\mathcal{P}(t,S,m)$ verifies $(P_1)$.
\end{lemma}

\begin{proof} The convergence of the learning process is guaranteed by the assumption that $B$ is computable in time $O(t)$. Let $f$ be the function computable in time $O(t)$ to which $\mathcal{P}(t,S,m)$ converges. This means that for almost all $n$ and all $x,y\in\{0,1\}$, if $S(n)=(x,y)$ then $f(x,y,n)=B(x,y,n)$, i.e.\ at least in the learning rounds, $f$ coincides with $B$ from some point on. Assume by contradiction that for infinitely many $n$
\begin{equation}\label{eq:lem2}
\exists x,y\in\{0,1\}.f(x,y,n)\neq B(x,y,n).
\end{equation}
Now, letting $g:\NN\to\{0,1\}$ be defined as $g(n)=1$ iff \eqref{eq:lem2} is true, and $\Gamma$ as $\{0,1\}^2$, we have by Fact \ref{fact-random} that $S$ is not $t$-random, a contradiction.
\end{proof}

\begin{lemma}\label{lemma-signaling}
If $S$ is $t$-random then $\mathcal{P}(t,S,m)$ verifies $(P_2)$.
\end{lemma}
\begin{proof}
  By assumption \ref{item:assumption} we have that for infinitely many $n$
    \begin{equation}\label{eq:lem3}
    \exists y\in\{0,1\}.B(0,y,n)\neq B(1,y,n).
    \end{equation}
Let $g:\NN\to\{0,1\}$ be defined as $g(n)=1$ iff \eqref{eq:lem3} is true, and $\Gamma$ as $\{1,\dots,m\}$. Assume by contradiction that for almost all $n$ we have that if $S(n)\in\Gamma$ then $g(n)=0$.
Now we have by Fact \ref{fact-random} that $S$ is not $t$-random, a contradiction.
\end{proof}

It is important to note that, without any knowledge of $B$, there is no a priori bound on the time it will take Bob to determine Alice's message with high enough confidence. 
Nonetheless, since this time is finite, there exists some finite distance for which the signaling allowed by our protocol is supraluminal. For instance, if it takes Bob $M$ rounds to find out Alice's message and each round takes a time $T$, then if they are at a distance $cTM$, the message is obtained before a light signal from Alice could reach Bob.

It could be argued that imposing a bound on the time complexity of Alice and Bob's boxes (which are nothing but an abstraction of what Nature is doing to choose the outputs) is a strong requirement. However, since the number of computational steps per second that can be performed by a system of mass $m$ is upper bounded by $2mc^2/\pi\hbar$ \cite{lloyd2000ultimate}, this is not only a requirement of our protocol but a reasonable physical assumption.


\textit{Discussion. --} Our protocol shows that correlated systems that would have violated a Bell inequality if were used for a standard Bell test, can be used to signal if assumed to be a computable and a time (or space) bound for their computational complexity is known in advance. The main consequence of this is that we are left with the following consequences: either Bell-violating systems cannot be computable, or if Alice and Bob guess properly a complexity class larger than the one used by the computable systems, they can signal in either way using the previous protocol.

The only assumptions to arrive at this result were the computable nature of the boxes and the requirement of violating a Bell inequality if used for such matter.


This work shows that in device independent scenarios, computability of results imposes a strong limitation on how nature can behave if it only had computable resources to generate outputs for the experiments.
Our result imply that, under the well established assumption that
no observable signaling exists, we need to accept the existance of truly unpredictable physical processes.

It is worth mentioning that our result doesn't go into conflict with the different interpretations of quantum mechanics. All of them predict random outputs, which are not allowed by our model. In the Copenhagen interpretation, the measurement process is postulated as random, whereas, for example, in Bohmian mechanics, it is deterministic but the initial conditions are randomly distributed and fundamentally unknowable.

This work was supported by the ERC CoG QITBOX, an AXA
Chair in Quantum Information Science, the Spanish MINECO (Project
FOQUS FIS2013-46768-P, Severo Ochoa grant
SEV-2015-0522 and FPI FIS2010-14830), grants ANPCyT-PICT-2013-2011, ANPCyT-PICT-2011-0365, UBACyT 20020110100025, the Laboratoire International Associ\'e ``INFINIS'', the Fundacion Cellex, the
Generalitat de Catalunya (SGR875),  and the John Templeton Foundation.

\bibliography{bibliography}

\end{document}